\documentclass[10pt,aps,pra,twocolumn,floatfix,nofootinbib,superscriptaddress,longbibliography]{revtex4-2}

\usepackage{graphicx} 
\usepackage{float}
\usepackage[]{xcolor}
\usepackage{booktabs}
\usepackage{amssymb}
\usepackage{amsmath}
\usepackage{amsthm}
\usepackage{mathtools}
\usepackage{tikz}
\usepackage{quantikz}
\usepackage{bm}
\usepackage{bbold} 
\usepackage{siunitx}
\usepackage{physics}
\usepackage[version=4]{mhchem}
\usepackage{enumitem}
\usepackage{tabularx}
\usepackage[ruled,lined,linesnumbered]{algorithm2e}
\usepackage{mdframed}
\usepackage{bibunits}

\definecolor{LEI-blue}{cmyk}{1,.75,0,.35} 
\definecolor{LEI-orange}{cmyk}{0,.62,.97,0} 
\definecolor{niceblue}{rgb}{.1, .25, .8}

\usepackage[colorlinks=true, linkcolor=niceblue, citecolor=gray]{hyperref}
\usepackage{orcidlink}

\newcommand{\topic}[1]{
}

\newtheorem{theorem}{Theorem}
\newtheorem{definition}[theorem]{Definition}
\newtheorem{lemma}[theorem]{Lemma}
\newtheorem{proposition}[theorem]{Proposition}

\DeclareMathOperator{\sgn}{sgn}

\newcommand{\clnum}{\operatorname{cl}}

\newcommand{\lorentz}[0]{%
	\affiliation{Instituut-Lorentz, Universiteit Leiden, P.O. Box 9506, 2300 RA Leiden, The Netherlands.}
}
\newcommand{\aQa}[0]{%
	\affiliation{$\langle \textit{aQa}^\textit{L} \rangle$ Applied Quantum Algorithms, Leiden University, The Netherlands.}
}

\newcommand{\cqt}[0]{%
	\affiliation{Centre for Quantum Technologies, National University of Singapore, 3 Science Drive 2, Singapore 117543, Singapore}
}

\newcommand{\hkust}[0]{%
	\affiliation{The Hong Kong University of Science and Technology, Clear Water Bay, Kowloon, Hong Kong.}
}

\newcommand{\padova}[0]{%
	\affiliation{Dipartimento di Fisica e Astronomia ``G. Galilei'', Università di Padova, I-35131 Padova, Italy}
	\affiliation{Padua Quantum Technologies Research Center, Università degli Studi di Padova, I-35131 Padova, Italy}
}


\newcommand{\Ftwo}{\mathbb F_2}
\newcommand{\R}{\mathbb R}
\newcommand{\Pcal}{\mathcal P}
\newcommand{\Mcal}{\mathcal M}
\newcommand{\Scal}{\mathcal S}
\newcommand{\STAB}{\mathrm{STAB}}
\newcommand{\stab}{\mathrm{stab}}
\newcommand{\RoM}{\mathrm{RoM}}
\newcommand{\conv}{\mathrm{conv}}
\renewcommand{\tr}{\operatorname{tr}}

\newcommand{\Ind}{\mathcal I}
\newcommand{\Hcal}{\mathcal H}
\newcommand{\Bset}{\mathrm B}

\providecommand{\mathbm}[1]{\bm{#1}}

\newtheorem*{propRelaxationRestatement}{Proposition~\ref{prop:relaxation}}
\newtheorem*{lemmaCollapseRestatement}{Lemma~\ref{lem:collapse}}
\newtheorem*{theoremMwisRestatement}{Theorem~\ref{thm:mwis}}
\newtheorem*{theoremSolvableRestatement}{Theorem~\ref{thm:solvable}}

\begin{document}
	
	\title{Graph Theoretic Approach to Quantum Nonstabilizerness}
	
	\author{Yingjian Liu
		\orcidlink{0009-0007-3872-1160}}
	\email{yingjian@lorentz.leidenuniv.nl}
	\lorentz
	\aQa
	
	\author{Albert Gasull
		\orcidlink{0000-0002-9272-6903}}
	\lorentz
	\aQa
	
	\author{Mengyao Hu
		\orcidlink{0000-0003-2621-3365}}
	\cqt
	
	\author{Ruiyun Zhang
		\orcidlink{0009-0001-4878-4372}}
	\hkust
	
	\author{Flavio Baccari
		\orcidlink{0000-0003-3374-5968}}
	\padova
	
	\author{Jordi Tura
		\orcidlink{0000-0002-6123-1422}}
	\lorentz
	\aQa
	
	\date{\today}
	
	\begin{abstract}
		Detecting nonstabilizerness requires full tomography and an optimization over exponentially many stabilizer states. A limited Pauli measurement set promises resource-efficient magic certification, yet the resulting reduced stabilizer polytope is generally difficult to characterize.
		We trace this difficulty into two coupled obstructions: the  simultaneous measurability of measurements captured by their frustration graph structure, and the consistency  of  sign dependencies from stabilizer formalism.
		We show that the sign dependencies can be discarded exactly whenever active dependencies are absent,  and that perfect frustration graphs then make this reduced polytope efficiently solvable.
		This solvable regime derives a closed form bounded by the clique number of the frustration graph, revealing a tradeoff between witness capacity and simultaneous measurability.
		Clifford covariance allows rotated measurement sets to enlarge the detectable state space without raising the capacity.
		Graph structure therefore emerges as both a certificate of tractability and a design principle for scalable magic resource detection.
	\end{abstract}
	
	\maketitle
	
	\emph{Introduction.}---
	\label{sec:intro}
	Nonstabilizerness, or magic, is the resource that enables quantum processors to go beyond the classically simulable stabilizer formalism~\cite{nielsen2010quantum}.
	Universal quantum computing requires non-Clifford resources supplied by non-Clifford gates or magic state injection~\cite{gottesman1997stabilizercodesquantumerror,shor1996faulttolerant,gottesman1999demonstrating,bravyi2005universal,wills2025constantoverhead,wang2019quantifyingmagicchannels,wang2020efficientbounds}.
	Nonstabilizerness admits many inequivalent quantifiers~\cite{veitch2012negative,howard2017application}, and efficient protocols are available for pure-state measures such as the stabilizer R\'enyi entropies~\cite{leone2022stabilizer,leone2024stabilizer,haug2023stabilizer}.
	
	However, for a general $n$-qubit quantum state $\rho$, evaluating standard monotones such as the robustness of magic (RoM) requires tomographically complete data and optimization over the $2^{\Theta(n^2)}$ vertices of the stabilizer polytope, since stabilizer polytope $\STAB_n$ is the convex hull of all $n$-qubit stabilizer states ~\cite{Veitch2014,garcia2017geometry,heinrich2019robustness,hamaguchi2024handbook,oliveirajunior2025geometric,leone2026unbearablehardnessdecidingmagic}.
	A recent measurement-limited framework addresses the memory bottleneck by projecting $\STAB_n$ onto the coordinates of a smaller chosen set $\Mcal$ of $m$ Pauli operators.
	This projection yields a reduced stabilizer polytope $\STAB(\Mcal)$ and the associated state monotone $\rho\mapsto\RoM_\Mcal(\rho)$~\cite{varela2026predictingmagicmeasurements}.
	Any measured correlation vector outside $\STAB(\Mcal)$ certifies magic, and $\RoM_\Mcal(\rho)$ lower bounds the full $\RoM(\rho)$.
	The frustration graph $G_\Mcal$ has one node per measured Pauli operator and an edge between each anticommuting pair, so its independent sets are pairwise commuting subsets of $\Mcal$, i.e., commuting Pauli operators that can be measured in a single experimental setting~\cite{chapman2020characterization}.
	The vertices of $\STAB(\Mcal)$ are indexed by maximal independent sets of $G_\Mcal$ together with sign patterns compatible with the Pauli product relations~\cite{cabello2014graph,chapman2020characterization,xu2024bounding,varela2026predictingmagicmeasurements}.
	This description exposes two structural hierarchies: the frustration graph determines simultaneously measurable Pauli operators, whereas product relations determine sign dependencies for enumerating vertices of $\STAB(\Mcal)$.
	
	Constrained by the two coupled hierarchies, deciding membership in the reduced polytope $\STAB(\Mcal)$ remains NP-hard for general measurement sets~\cite{varela2026predictingmagicmeasurements}.
	To address this issue, we relax $\STAB(\Mcal)$ by allowing all sign assignments on each maximal independent set, thereby discarding only the restrictions imposed by Pauli product relations.
	We prove that this relaxation is exact if and only if $\Mcal$ has no active dependencies, where an active dependency is a minimal Pauli product relation contained entirely in an independent set, as depicted schematically in Fig.~\ref{fig:graph-polytope-rom}.
	Under this relaxation, the separation oracle for the dual constraints reduces to a standard maximum-weight independent set (MWIS) problem on $G_\Mcal$.
	Exploiting the polynomial-time solvability of MWIS on perfect graphs, the evaluation of $\RoM_\Mcal(\rho)$ is further reduced to an efficiently computable closed form in the absence of active dependencies.
	
	Beyond evaluating $\RoM_\Mcal(\rho)$, the closed form determines the witness capacity of a measurement set $\Mcal$, defined as $\sup_\rho\RoM_\Mcal(\rho)$.
	In the solvable regime, the witness capacity equals $\sqrt{\clnum(G_\Mcal)}$, where the clique number $\clnum(G_\Mcal)$ is the size of the largest pairwise anticommuting subset of $\Mcal$.
	The closed form therefore quantifies the trade-off between witness capacity and simultaneous measurability.
	
	A high witness capacity alone does not guarantee broad detection because a single reduced stabilizer polytope has limited state-space coverage.
	We therefore prove that Clifford unitary rotations preserve the solvability and witness capacity of the reduced stabilizer polytope, while covering different Pauli directions to enlarge the detected set.
	The numerical results certify nonstabilizerness of Haar-random states and hardware-efficient variational states using Clifford rotations of dedicated Pauli measurement sets.
	Graph structure is therefore not only a certificate of tractability but also a practical design principle for measurement sets.
	The same viewpoint extends beyond the solvable regime and connects the linear reduced witness to quadratic and entropic diagnostics.
	
	In summary, our contributions are fourfold: (1) We introduce the sign-relaxed reduced stabilizer polytope and establish the condition for its exactness in Theorem~\ref{thm:mwis}. (2) We identify its dual separation oracle as MWIS on $G_\Mcal$ in Proposition~\ref{prop:relaxation}. (3) We derive a closed form for reduced RoM and characterize its witness capacity in Theorem~\ref{thm:solvable}. (4) We use Clifford covariance in Proposition~\ref{prop:clifford} to numerically compare the detectability of different measurement sets.
	
	\emph{Background.}---
	\label{sec:background}
	Evaluating the robustness of magic over $2^n\prod_{j=1}^{n}(2^j+1)\approx2^{\Theta(n^2)}$ vertices of $\STAB_n$ requires tomographically complete data and a classically intractable optimization.
	Instead, one can fix a measurement set $\Mcal=\{P_1,\dots,P_m\}$ of nonidentity Pauli operators modulo phase, so that the experiment reports only the accessible expectation vector $\mathbm{b}_{\Mcal(\rho)}=(\tr(P_1\rho),\dots,\tr(P_m\rho))\in[-1,1]^m$.
	Seen through this window, stabilizer states are projected into the reduced stabilizer polytope
	{\small
		\begin{equation}
			\STAB(\Mcal)=\conv\{\mathbm{b}_{\Mcal(\sigma)}:\sigma\in\STAB_n\}\subset\R^m,
			\label{eq:reduced-stab-polytope}
		\end{equation}
	}
	and the reduced robustness of magic is the induced monotone
	{\small
		\begin{equation}
			\RoM_\Mcal(\rho)
			=\min_{\mathbm{x}}\!\Big\{\|\mathbm{x}\|_1:\!\textstyle\sum_v x_v\mathbm{v}=\mathbm{b}_{\Mcal(\rho)},\ \textstyle\sum_v x_v=1\Big\},
			\label{eq:reduced-primal}
		\end{equation}
	}
	where the sum runs over the vertices $\mathbm{v}$ of $\STAB(\Mcal)$. 
	The projection gives $\RoM_\Mcal(\rho)\le\RoM(\rho)$, hence $\RoM_\Mcal(\rho)>1$ certifies nonstabilizerness.

	The frustration graph $G_\Mcal=(\Mcal,E_\Mcal)$ has node set $\Mcal$ and an edge $\{P_i,P_j\}$ for anticommuting pairs $P_iP_j=-P_jP_i$~\cite{chapman2020characterization,xu2024bounding}.
	An independent set contains no adjacent nodes and therefore corresponds to a pairwise commuting context $S\subseteq\Mcal$.
	Write $\mathsf{Sign}(S):=\{f:S\to\{\pm1\}\}$ for the set of all sign assignments on $S$.
	A pure stabilizer state $\sigma$ is a joint eigenstate of $S$ with deterministic projected coordinates $f(P)=\tr(P\sigma)\in\{\pm1\}$ for every $P\in S$.
	These signs are constrained by Pauli product relations: if $\prod_{P\in T}P=\eta_T\mathbb{1}$ for a $T\subseteq S$, then $\prod_{P\in T}f(P)=\eta_T$.
	Accordingly, the sign assignments satisfying all such relations form the admissible sign set $\Bset_S\subseteq\mathsf{Sign}(S)$.
	For $f\in\mathsf{Sign}(S)$ or $f\in\Bset_S$, define $\mathbm{v}_{S,f}\in\R^m$ by $(\mathbm{v}_{S,f})_i=f(P_i)$ if $P_i\in S$ and $(\mathbm{v}_{S,f})_i=0$ otherwise.
	We write $\Ind_{\max}(G_\Mcal)$ for the maximal independent sets of $G_\Mcal$, equivalently the maximal commuting contexts.
	The reduced stabilizer polytope then has the exact V-representation~\cite{varela2026predictingmagicmeasurements}
	{\small
		\begin{equation}
			\STAB(\Mcal)=\conv\{\mathbm{v}_{S,f}:S\in\Ind_{\max}(G_\Mcal),\ f\in\Bset_S\}.
			\label{eq:exact-vrep-new}
		\end{equation}
	}
	Thus, the graph fixes the maximal commuting supports $S$, while $\Bset_S$ selects the sign patterns compatible with Pauli product relations.
	The two hierarchies meet in the linear programming dual of Eq.~\eqref{eq:reduced-primal}, whose constraints range over $S\in\Ind_{\max}(G_\Mcal)$ and $f\in\Bset_S$~\cite{hamaguchi2024handbook}
	{\small
		\begin{equation}
			\RoM_\Mcal(\rho)=\max_{\mathbm{y},\mu}\Big\{\mathbm{b}_{\Mcal(\rho)}^\top\mathbm{y}+\mu:\ \big|\mathbm{v}_{S,f}^\top\mathbm{y}+\mu\big|\le1,\ \forall S,f\Big\}.
			\label{eq:dual-exact}
		\end{equation}
	}
	Checking the dual constraints requires maximizing over $f\in\Bset_S$. Since $\Bset_S$ is a coset of a binary linear code, the two branches of the absolute value lead to linear optimizations over this affine code. Nevertheless, deciding membership in $\STAB(\Mcal)$ is NP-hard for a general Pauli measurement set~\cite{varela2026predictingmagicmeasurements}.

	\begin{figure}[t]
		\centering
		\definecolor{figactive}{HTML}{C65911}
		\definecolor{figinactive}{HTML}{6AAE94}
		\definecolor{figrom}{HTML}{00A6D6}
		\definecolor{figromrelaxed}{HTML}{215E9A}
		\begin{tikzpicture}
			\node[anchor=south west,inner sep=0] (polytopefig) at (0,0) {\includegraphics[width=0.92\columnwidth]{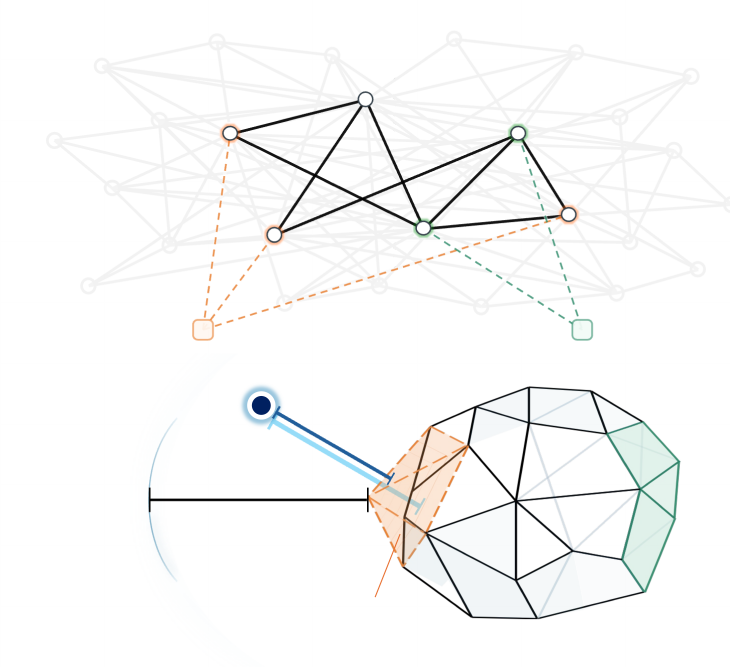}};
			\begin{scope}[x={(polytopefig.south east)},y={(polytopefig.north west)}]
				\node[font=\footnotesize] at (0.639,0.838) {$G_\Mcal$};
				\node[font=\footnotesize,text=figactive] at (0.205,0.530) {active};
				\node[font=\footnotesize,text=figinactive] at (0.913,0.520) {inactive};
				\node[font=\footnotesize] at (0.304,0.413) {$\rho$};
				\node[font=\scriptsize,text=figromrelaxed,rotate=-30.6] at (0.457,0.366) {$\widetilde{\RoM}_{\Mcal}$};
				\node[font=\scriptsize,text=figrom,rotate=-34.0] at (0.420,0.304) {$\RoM_{\Mcal}$};
				\node[font=\scriptsize] at (0.357,0.208) {$\sqrt{\clnum(G_\Mcal)}$};
				\node[font=\tiny,text=figactive] at (0.485,0.079) {$\widetilde{\STAB}(\Mcal)\setminus\STAB(\Mcal)$};
				\node[font=\footnotesize] at (0.892,0.037) {$\STAB(\Mcal)$};
			\end{scope}
		\end{tikzpicture}
		\caption{Schematic frustration graph, dependency checks, and exact versus sign-relaxed reduced stabilizer polytopes. An active dependency removes the orange sector of $\widetilde{\STAB}(\Mcal)$ to yield $\STAB(\Mcal)$, whereas an inactive dependency leaves the green boundary unchanged. The cyan and navy lines schematically denote $\RoM_\Mcal(\rho)$ and $\widetilde{\RoM}_\Mcal(\rho)$, respectively. If the frustration graph is a perfect graph, the reduced monotone is upper bounded by $\sqrt{\clnum(G_\Mcal)}$. }
		\label{fig:graph-polytope-rom}
	\end{figure}
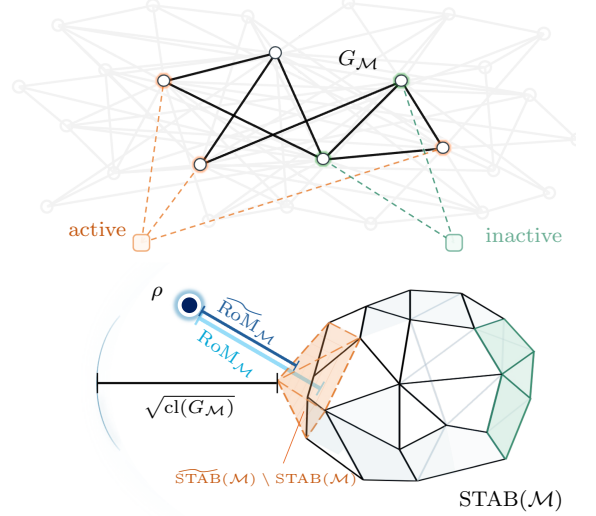

	\emph{Results.}---
	\label{sec:results}
	Enforcing the admissible-sign constraints is generally hard.
	We therefore relax the sign dependencies by retaining all deterministic sign assignments for a sign-relaxed stabilizer polytope
	{\small
		\begin{equation}
			\widetilde{\STAB}(\Mcal)
			:=\conv\{\mathbm{v}_{S,f}:S\in\Ind_{\max}(G_\Mcal),\ f\in\mathsf{Sign}(S)\}.
			\label{eq:relaxed-polytope}
		\end{equation}
	}
	This construction is an outer relaxation satisfying $\STAB(\Mcal)\subseteq\widetilde{\STAB}(\Mcal)$.
	The associated relaxed monotone for a state is
	{\small
		\begin{equation}
			\widetilde{\RoM}_\Mcal(\rho)=\min_{\mathbm{x}}\!\Big\{\|\mathbm{x}\|_1:\!\textstyle\sum_{\tilde v}x_{\tilde v}\tilde{\mathbm{v}}=\mathbm{b}_{\Mcal(\rho)},\ \textstyle\sum_{\tilde v}x_{\tilde v}=1\Big\},
			\label{eq:relaxed-rom}
		\end{equation}
	}
	where the sums run over the vertices of $\widetilde{\STAB}(\Mcal)$.
	Consequently,
	\begin{equation}
		\widetilde{\RoM}_\Mcal(\rho)\le\RoM_\Mcal(\rho).
		\label{eq:relaxed-lower-bound}
	\end{equation}
	Detailed proofs of the following statements and the supporting graph theoretic results are given in Appendix~\ref{app:proof-mwis}.
	\begin{definition}[Active dependency]
		\label{def:active}
		A dependency is a nonempty inclusion-minimal subset $T\subseteq\Mcal$ satisfying $\prod_{P\in T}P\propto\mathbb{1}$ for some, equivalently any, ordering.
		It is active if its Pauli operators pairwise commute, i.e. $T\in\Ind(G_\Mcal)$.
	\end{definition}
	Reordering changes the product only by a sign, which is absorbed by $\propto$. For example, $\Mcal=\{X_1X_2,Z_1Z_2,Y_1Y_2\}$ is commuting and its three operators multiply to $-\mathbb{1}$. The exact signs therefore obey one parity check and realize four of the eight deterministic patterns retained by the relaxation, making the inclusion strict.
	\begin{lemma}
		\label{lem:collapse}
		$\Mcal$ has no active dependencies if and only if $\Bset_S=\mathsf{Sign}(S)$ for every commuting context $S\in\Ind(G_\Mcal)$.
	\end{lemma}
	For an active dependency, the product has phase $\pm1$ and imposes a parity check on the admissible signs.
	Thus $\Mcal$ has no active dependencies when every dependency contains an anticommuting pair. A simple example is the single-qubit Pauli set $\Mcal = \{X, Y, Z\}$.
	\begin{theorem}
		
		\label{thm:mwis}
		$\STAB(\Mcal)=\widetilde{\STAB}(\Mcal)$ holds if and only if $\Mcal$ has no active dependencies.
		When these equivalent conditions hold, for every quantum state $\rho$,
		\begin{equation}
			\RoM_\Mcal(\rho)=\widetilde{\RoM}_\Mcal(\rho).
			\label{eq:no-active-tight}
		\end{equation}
		Consequently, the polytope relaxation computes the exact reduced robustness without a correction from sign dependencies.
	\end{theorem}
	The sign relaxation exposes a graph-only constraint in the dual of reduced RoM. To show this, we recall the associated weighted optimization:
	For nonnegative node weights $\mathbm{w}\in\mathbb{R}_{\ge0}^{\Mcal}$, the maximum-weight independent set (MWIS) value is defined as~\cite{karp1972reducibility}
	\begin{equation}
		\alpha_{\mathbm{w}}(G_\Mcal):=\max_{I\in\Ind(G_\Mcal)}\sum_{P_i\in I}w_i .
		\label{eq:wmis}
	\end{equation}
	Here $\Ind(G_\Mcal)$ denotes the set of all independent sets of $G_\Mcal$. 
	Recent research shows that energy witnesses reduce a single Hamiltonian support function to MWIS under restrictive independence assumptions~\cite{macedo2026heat}. We extend this perspective to the reduced stabilizer polytope and identify the solvable regime below.
	\begin{proposition}
		
		\label{prop:relaxation}
		For any Pauli measurement set $\Mcal=\{P_1,\dots,P_m\}$ and a quantum state $\rho$,
		{\small
			\begin{equation}
				\widetilde{\RoM}_\Mcal(\rho)=\max_{\mathbm{y},\mu}\Big\{\mathbm{b}_{\Mcal(\rho)}^\top \mathbm{y}+\mu: \alpha_{|\mathbm{y}|}(G_\Mcal)+|\mu|\le 1\Big\},
				\label{eq:dual-relaxed}
			\end{equation}
		}
		where $|\mathbm{y}|:=(|y_1|,\dots,|y_m|)^\top$. Testing a candidate dual point requires evaluating the MWIS value $\alpha_{|\mathbm{y}|}(G_\Mcal)$ defined in Eq.~\eqref{eq:wmis}. 
	\end{proposition}
	Within $G_\Mcal$, a clique is a subset $Q\subseteq\Mcal$ whose nodes are pairwise adjacent, equivalently a pairwise anticommuting set of Pauli operators. Its maximum cardinality is the clique number $\clnum(G_\Mcal)$.
	A proper coloring assigns colors to nodes such that adjacent nodes have different colors, and the chromatic number $\chi(G_\Mcal)$ is the minimum number of colors required. For a frustration graph, each color class is a commuting context, so $\chi(G_\Mcal)$ is the fewest commuting contexts that partition $\Mcal$.
	The nodes of a clique require distinct colors, and hence $\clnum(G_\Mcal)\le\chi(G_\Mcal)$. A graph $G$ is perfect if $\chi(G')=\clnum(G')$ for every induced subgraph $G'\subseteq G$~\cite{chudnovsky2006strong}.
	Perfect graphs admit polynomial-time MWIS optimization, and in this regime the graph-only dual further yields the closed form below~\cite{lovasz1972normal,lovasz1972characterization,lovasz1979shannon,grotschel1981ellipsoid,grotschel1988geometric,chvatal1975certain}.
	
	\begin{theorem}
		\label{thm:solvable}
		If $\Mcal$ has no active dependencies and $G_\Mcal$ is perfect, then for any $n$-qubit quantum state $\rho$,
		{\small
			\begin{equation}\label{eq:closed-form}
				\RoM_\Mcal\big(\rho\big)
				=
				\max\!\Big(1,\ \max_{Q} \sum_{P\in Q}\big|\tr(P\rho)\big| \Big)
				\le\sqrt{\clnum(G_\Mcal)},
			\end{equation}
		}
		where $Q$ ranges over cliques of $G_\Mcal$, equivalently over pairwise anticommuting subsets of $\Mcal$.
		The bound is attained by any state supported on the $+1$ eigenspace of $|Q|^{-1/2}\sum_{P\in Q}P$ for the maximum clique $Q$.
	\end{theorem}
	For finite precision expectation values, evaluating Eq.~\eqref{eq:closed-form} is a maximum-weight clique problem with weights $|\tr(P\rho)|$, equivalently MWIS on complement graph $\overline{G_\Mcal}$. The complement of a perfect graph is perfect, so Lemma~\ref{lem:perfect-collapse} in Appendix~\ref{app:proof-mwis} gives a polynomial-time evaluation.
	A proper coloring of $G_\Mcal$ partitions $\Mcal$ into mutually commuting Pauli families, which can be measured jointly after Clifford diagonalization. 
	Since a perfect frustration graph satisfies $\chi(G_\Mcal)=\clnum(G_\Mcal)$, all expectation values required by Eq.~\eqref{eq:closed-form} can be collected using at most $\clnum(G_\Mcal)$ commuting settings, instead of original $m$ different settings. 
	Moreover, $\clnum(G_\Mcal)\le2n+1$ for $n$-qubit Pauli operators, giving the universal ceiling $\RoM_\Mcal(\rho)\le\sqrt{2n+1}$ throughout the solvable regime~\cite{sarkar2021sets}. 
	
	The ceiling is reached when the frustration graph realizes the maximum anticommuting clique permitted.
	Fermionic system provides an example for constructing such measurement sets.  Note that the $n$-mode creation and annihilation operators $a_k^\dagger$ and $a_k$ define $2n$ Hermitian Majorana generators $c_{2k-1}=a_k+a_k^\dagger$ and $c_{2k}=i(a_k^\dagger-a_k)$, satisfying $\{c_i,c_j\}=2\delta_{ij} \mathbb{1}$.
	Under the Jordan--Wigner transformation, they become the desired Pauli strings $\gamma_{2k-1}=Z_1\cdots Z_{k-1}X_k$ and $\gamma_{2k}=Z_1\cdots Z_{k-1}Y_k$.
	In addtion, the global parity operator $\gamma_{2n+1}=Z_1\cdots Z_n=(-i)^n\gamma_1\cdots\gamma_{2n}$ anticommutes with all of them, completing the Jordan--Wigner(JW) set $\Gamma=\{\gamma_i\}_{i=1}^{2n+1}$ of pairwise anticommuting Pauli operators~\cite{brauer1935spinors,ipek2026phasespace}.
	Then $G_\Gamma=K_{2n+1}$, the complete graph $K_{2n+1}$ on $2n+1$ nodes is perfect and satisfies $\clnum(G_\Gamma) = 2n+1$. 
	Its unique nontrivial dependency $\prod_i\gamma_i\propto\mathbb{1}$ includes anticommuting pairs, making sure that the JW set has no active dependencies. 
	The reduced stabilizer polytope of $K_{2n+1}$ is a cross-polytope, while anticommuting uncertainty bounds quantum expectation vectors by the Euclidean ball~\cite{cabello2014graph,xu2024bounding,xu2025simultaneous}.
	
	Majorana generators also construct the tree-structured measurement set $\Mcal_{\rm tree}$ defined in Table~\ref{tab:measurement-sets}. This dependency-free construction has a chordal, hence perfect, frustration graph and provides another solvable family whose ceiling is $\sqrt{\lfloor3n/2\rfloor}$. 	
	Figure~\ref{fig:measurement-graphs} and Table~\ref{tab:measurement-sets} summarize four solvable measurement sets, their frustration graphs, and their witness capacities.

	\begin{proposition}[Clifford covariance]
		\label{prop:clifford}
		Let $\Mcal=\{P_i\}_{i=1}^m$ be any Pauli measurement set and let $C$ be a Clifford unitary.
		With $C\Mcal C^\dagger=\{CP_iC^\dagger\}_{i=1}^m$ ordered by $P_i\mapsto CP_iC^\dagger$, one has
		\begin{equation}
			\RoM_{C\Mcal C^\dagger}(\rho)=\RoM_\Mcal(C^\dagger\rho C).
			\label{eq:clifford}
		\end{equation}
		If $\Mcal$ has no active dependencies and $G_\Mcal$ is perfect, then $C\Mcal C^\dagger$ has the same property and the closed form becomes
		{\small
			\begin{equation}
				\RoM_{C\Mcal C^\dagger}(\rho)
				=
				\max\!\Big(1,\ \max_Q\sum_{P_i\in Q}\big|\tr(CP_iC^\dagger\rho)\big|\Big),
				\label{eq:clifford-closed-form}
			\end{equation}
		}
		where $Q$ ranges over cliques of $G_\Mcal$, and $CQC^\dagger$ ranges over cliques of $G_{C\Mcal C^\dagger}$.
	\end{proposition}
	
	\emph{Proof sketch.} Clifford conjugation maps stabilizer states bijectively to stabilizer states and gives $\mathbm{b}_{(C\Mcal C^\dagger)(\rho)}=\mathbm{b}_{\Mcal(C^\dagger\rho C)}$.
	The two linear programs are therefore identical after the coordinate identification.
	Although $CP_iC^\dagger$ may differ by a sign from a chosen phase-free representative, the conjugated observable itself is retained, so this coordinate identity is unambiguous.
	Clifford conjugation also preserves commuting, anticommuting, and Pauli product relations, hence both $G_\Mcal$ and the active status of every dependency.
	Thus, for a solvable seed measurement set, any ensemble witness $\max_\ell\RoM_{C_\ell\Mcal C_\ell^\dagger}(\rho)$ obeys the same witness-capacity bound $\sqrt{\clnum(G_\Mcal)}$.
	
	\newcommand{\mainCliffordNumerics}{%
		\begin{figure}[!t]
			\centering
			\includegraphics[width=\columnwidth]{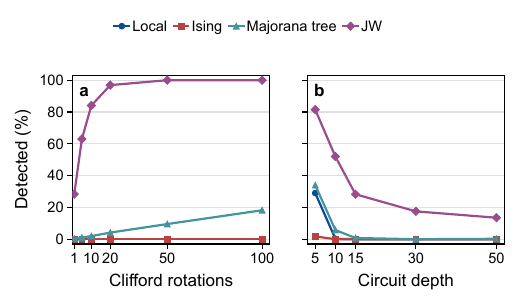}
			\caption{
				$8$-qubit magic-detection rates for the Local, Ising, Majorana-tree, and JW measurement sets. (\textbf{a}) Haar-random pure states versus the cumulative number of global Clifford rotations. (\textbf{b}) hardware-efficient variational states versus circuit depth.
			}
			\label{fig:detection-numerics}
		\end{figure}
	}
	\mainCliffordNumerics
	
	\newcommand{\mainMeasurementDefinitions}{%
		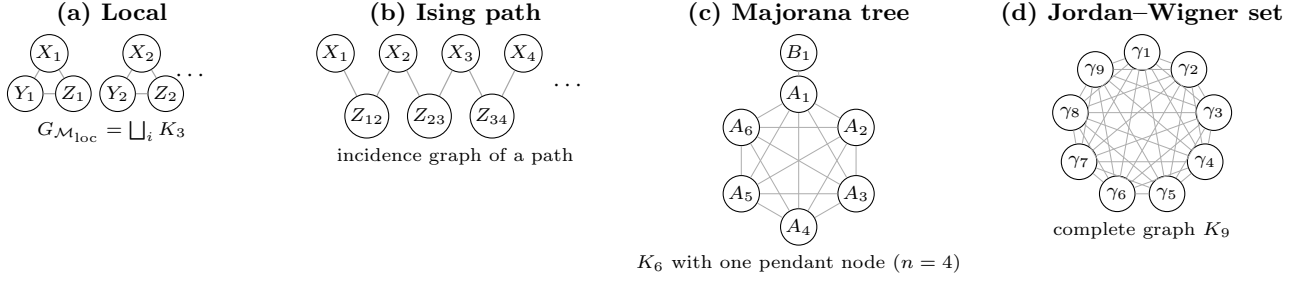
\begin{figure*}[!t]
			\centering
			\begin{minipage}[t]{0.24\textwidth}
				\centering
				\textbf{(a) Local}\\[2pt]
				\begin{tikzpicture}[
					scale=0.74,
					pnode/.style={circle,draw,inner sep=0.8pt,minimum size=4.7mm,font=\scriptsize},
					fedge/.style={draw=gray!70,line width=0.45pt}
					]
					\foreach \q/\x in {1/0,2/1.65}{
						\node[pnode] (lx\q) at (\x,0.72) {$X_{\q}$};
						\node[pnode] (ly\q) at (\x-0.43,0) {$Y_{\q}$};
						\node[pnode] (lz\q) at (\x+0.43,0) {$Z_{\q}$};
						\draw[fedge] (lx\q) -- (ly\q) -- (lz\q) -- (lx\q);
					}
					\node at (2.55,0.32) {$\cdots$};
				\end{tikzpicture}\\[-2pt]
				{\scriptsize $G_{\Mcal_{\rm loc}}=\bigsqcup_i K_3$}
			\end{minipage}
			\hfill
			\begin{minipage}[t]{0.24\textwidth}
				\centering
				\textbf{(b) Ising path}\\[2pt]
				\begin{tikzpicture}[
					scale=0.72,
					pnode/.style={circle,draw,inner sep=0.8pt,minimum size=4.7mm,font=\scriptsize},
					znode/.style={circle,draw,inner sep=0.45pt,minimum size=5.8mm,font=\scriptsize},
					fedge/.style={draw=gray!70,line width=0.45pt}
					]
					\foreach \i/\x in {1/0,2/1.15,3/2.30,4/3.45}{
						\node[pnode] (ix\i) at (\x,1.15) {$X_{\i}$};
					}
					\foreach \i/\j/\x/\lab in {1/2/0.575/Z_{12},2/3/1.725/Z_{23},3/4/2.875/Z_{34}}{
						\node[znode] (izz\i) at (\x,0) {$\lab$};
						\draw[fedge] (izz\i) -- (ix\i);
						\draw[fedge] (izz\i) -- (ix\j);
					}
					\node at (4.30,0.55) {$\cdots$};
				\end{tikzpicture}\\[-2pt]
				{\scriptsize incidence graph of a path}
			\end{minipage}
			\hfill
			\begin{minipage}[t]{0.24\textwidth}
				\centering
				\textbf{(c) Majorana tree}\\[2pt]
				\begin{tikzpicture}[
					scale=0.73,
					pnode/.style={circle,draw,inner sep=0.65pt,minimum size=4.8mm,font=\scriptsize},
					fedge/.style={draw=gray!64,line width=0.38pt}
					]
					\foreach \i/\ang in {1/90,2/30,3/-30,4/-90,5/-150,6/150}{
						\node[pnode] (a\i) at (\ang:1.20) {$A_{\i}$};
					}
					\node[pnode] (b1) at (90:1.95) {$B_1$};
					\foreach \i/\j in {1/2,1/3,1/4,1/5,1/6,2/3,2/4,2/5,2/6,3/4,3/5,3/6,4/5,4/6,5/6}{
						\draw[fedge] (a\i) -- (a\j);
					}
					\draw[fedge] (a1) -- (b1);
				\end{tikzpicture}\\[-2pt]
				{\scriptsize $K_6$ with one pendant node ($n=4$)}
			\end{minipage}
			\hfill
			\begin{minipage}[t]{0.24\textwidth}
				\centering
				\textbf{(d) Jordan--Wigner set}\\[2pt]
				\begin{tikzpicture}[
					scale=0.77,
					pnode/.style={circle,draw,inner sep=0.8pt,minimum size=4.7mm,font=\scriptsize},
					fedge/.style={draw=gray!64,line width=0.34pt}
					]
					\foreach \i/\ang in {1/90,2/50,3/10,4/-30,5/-70,6/-110,7/-150,8/170,9/130}{
						\node[pnode] (g\i) at (\ang:1.25) {$\gamma_{\i}$};
					}
					\foreach \i/\j in {1/2,1/3,1/4,1/5,1/6,1/7,1/8,1/9,2/3,2/4,2/5,2/6,2/7,2/8,2/9,3/4,3/5,3/6,3/7,3/8,3/9,4/5,4/6,4/7,4/8,4/9,5/6,5/7,5/8,5/9,6/7,6/8,6/9,7/8,7/9,8/9}{
						\draw[fedge] (g\i) -- (g\j);
					}
				\end{tikzpicture}\\[-2pt]
				{\scriptsize complete graph $K_9$}
			\end{minipage}
			\caption{
				Representative frustration graph structures for four solvable Pauli measurement sets.
				\textbf{(a)}. Local XYZ family for arbitrary $n$, with the ellipsis denoting further disjoint $K_3$ blocks.
				\textbf{(b)}. Open Ising path for arbitrary $n$, with the ellipsis denoting continuation of the path.
				\textbf{(c)}. Four-qubit Majorana tree, whose frustration graph is $K_6$ with one pendant node.
				\textbf{(d)}. Four-qubit Jordan--Wigner set $\Gamma_4=\{\gamma_1,\ldots,\gamma_9\}$ realizing $K_9$. In general, $\Gamma_n$ realizes $K_{2n+1}$.
			}
			\label{fig:measurement-graphs}
		\end{figure*}
		
		\begin{table*}[!t]
			\caption{
				Pauli definitions and witness capacities for the four measurement sets in Fig.~\ref{fig:measurement-graphs}.
				The witness capacity is $\sup_\rho\RoM_\Mcal(\rho)=\sqrt{\clnum(G_\Mcal)}$.
				The first $2n$ operators in $\Gamma_n$ are the standard Jordan--Wigner Majoranas, and $\gamma_{2n+1}$ completes the JW set~\cite{zurel2025simulation}.
			}
			\label{tab:measurement-sets}
			\small
			\setlength{\tabcolsep}{4pt}
			\begin{tabular}{@{}llc@{}}
				\toprule
				Measurement set & \multicolumn{1}{c}{Pauli operators} & Witness capacity \\
				\midrule
				Local XYZ ($n$ qubits) & \parbox[t]{0.56\textwidth}{$\Mcal_{\rm loc}=\{X_i,Y_i,Z_i:\ i=1,\ldots,n\}$.} & $\sqrt3$ \\[2pt]
				Ising path ($n\ge2$ qubits) & \parbox[t]{0.56\textwidth}{$\Mcal_{\rm Ising}=\{X_i:\ i=1,\ldots,n\}\cup\{Z_iZ_{i+1}:\ i=1,\ldots,n-1\}$.} & $\sqrt2$ \\[2pt]
				Majorana tree ($n\ge2$ qubits) & \parbox[t]{0.56\textwidth}{$\Mcal_{\rm tree}=\{A_a=i\gamma_1\gamma_{a+1}\}_{a=1}^{d_n}\cup\{B_a=i\gamma_{a+1}\gamma_{d_n+a+1}\}_{a=1}^{r_n}$, where $d_n=\lfloor3n/2\rfloor$ and $r_n=\lceil n/2\rceil-1$.} & $\sqrt{\lfloor3n/2\rfloor}$ \\[2pt]
				JW set ($n$ qubits) & \parbox[t]{0.56\textwidth}{$\gamma_{2k-1}=Z_1\cdots Z_{k-1}X_k$ and $\gamma_{2k}=Z_1\cdots Z_{k-1}Y_k$ for $k=1,\ldots,n$, together with $\gamma_{2n+1}=Z_1\cdots Z_n$.} & $\sqrt{2n+1}$ \\
				\bottomrule
			\end{tabular}
		\end{table*}
	}

	\mainMeasurementDefinitions
	Adding inequivalent Clifford rotations can enlarge the detected set without increasing this bound, while Clifford symmetries of a measurement set identify redundant rotations~\cite{nation2026clifford}. Figure~\ref{fig:detection-numerics} compares the empirical detection rates of four measurement sets at $n=8$.
	The detection rate is the fraction of states for which the largest reduced robustness over the evaluated measurement-set copies exceeds one.
	Fig.~\ref{fig:detection-numerics} \textbf{a} tests a common batch of $10{,}000$ Haar-random states by $K\in\{1,5,10,20,50,100\}$ random Clifford rotations.
	For the variational-state scan in Fig.~\ref{fig:detection-numerics} \textbf{b}, a common batch of $400$ states is evaluated at each circuit depth $L\in\{5,10,15,30,50\}$, and each circuit block consists of independent single-qubit $R_Y$ and $R_Z$ rotations followed by a periodic ring of controlled-$Z$ gates.
	The numerical result confirms the coverage gain from Clifford rotations, most clearly for the JW and Majorana-tree sets, while both panels show the broader empirical trend that larger witness capacity is associated with higher detection rates.
	
	The full Pauli measurement set exhibits both obstructions. Its frustration graph is the symplectic graph together with the isolated identity node. For $n\ge2$ it is imperfect, and its chromatic number $2^n+1$ grows exponentially, whereas its clique number is only $2n+1$~\cite{brouwer2012spectra}. After fixing the identity sign, each maximal commuting context also carries $2^n-1-n$ independent constraints imposed by sign dependencies~\cite{varela2026predictingmagicmeasurements}. Nevertheless, the same graph organizes nonlinear diagnostics beyond reduced robustness. The squared Pauli expectation profile connects two complementary nonlinear diagnostics: weighted graph functionals can yield quadratic magic witnesses, whereas its global concentration underlies stabilizer R\'enyi entropy~\cite{xu2025simultaneous,leone2022stabilizer}. Both constructions use the independent-set supports of stabilizer states as a common graph-theoretic reference.
	Appendix~\ref{app:sre} develops these connections and discusses their extension and limitations for mixed states. Thus graph support constraints and Pauli dependency codes provide a common language for exactly computable reduced witnesses and broader diagnostics of nonstabilizerness.
	
	\emph{Outlook.}---
	\label{sec:outlook}
	Our work identifies an exactly solvable core and separates the remaining difficulty into sign dependencies and graph structure. When $G_\Mcal$ is perfect but active dependencies are present, as in Ising models on graphs with cycles, MWIS remains tractable while the admissible signs require coset decoding of the dependency code. Characterizing dependency codes that admit efficient decoding is left for future work. Furthermore, a complementary direction concerns imperfect frustration graphs without active dependencies. Replacing $\alpha_{|\mathbm{y}|}(G_\Mcal)$ in Eq.~\eqref{eq:dual-relaxed} by the weighted Lov\'asz bound gives a sound polynomial-time lower bound on $\RoM_\Mcal(\rho)$. Determining which imperfect graphs keep this semidefinite bound tight is also an open problem.
	
	\emph{Acknowledgments.}---
	The authors acknowledge discussions about graph theory with Eloïc Vallée, and viewpoints from quantum nonlocality with Jin-fu Chen.  J.T. acknowledges support from the European Union's Horizon Europe research and innovation programme through the ERC StG FINE-TEA-SQUAD (Grant No.~101040729).
	J.T. also acknowledges support from the Dutch National Growth Fund (NGF), as part of the Quantum Delta NL programme.
	This work is part of the ``Quantum Inspire, the Dutch Quantum Computer in the Cloud'' project (project number NWA.1292.19.194) of the NWA research programme ``Research on Routes by Consortia (ORC)'', which is funded by the Netherlands Organization for Scientific Research (NWO). M.H. acknowledges the support from the National
	Research Foundation, Singapore through the National
	Quantum Office, hosted in A*STAR, under its Centre for
	Quantum Technologies Funding Initiative (S24Q2d0009). F.B. acknowledges financial support from the European Union's Horizon Europe research and innovation programme under the Marie Sk{\l}odowska--Curie Action for the project No.~101148556 (ENCHANT).
	The views and opinions expressed here are solely those of the authors and do not necessarily reflect those of the funding institutions.
	The funding institutions cannot be held responsible for them.
	
	\makeatletter
	\immediate\write\@auxout{\string\citation{BFLongBibliographyControl}}
	\makeatother
	\bibliographystyle{apsrev4-2}
	\bibliography{references}
	\onecolumngrid
	\newpage
	\appendix
	\begin{bibunit}[apsrev4-2]
		\makeatletter
		\def\@bibunitname{appendix}
		\let\hyper@natlinkstart\@gobble
		\let\hyper@natlinkend\relax
		\def\hyper@natlinkbreak#1#2{#1}
		\let\hyper@natanchorstart\@gobble
		\let\hyper@natanchorend\relax
		\makeatother
		
		\section{Reduced polytope structure}
		\label{app:binary-symplectic}
		
		For a Pauli measurement set $\Mcal$, the reduced stabilizer polytope $\STAB(\Mcal)$ of Eq.~\eqref{eq:reduced-stab-polytope} is the projection of the full stabilizer polytope onto the measured coordinates, and $\RoM_\Mcal(\rho)$ of Eq.~\eqref{eq:reduced-primal} is a stabilizer monotone with $\RoM_\Mcal(\rho)\le\RoM(\rho)$~\cite{varela2026predictingmagicmeasurements}.
		$\STAB(\Mcal)$ has the combinatorial $V$-representation
		\begin{equation}
			\begin{aligned}
				\STAB(\Mcal)&=\conv\{\mathbm{v}_{S,f}\}_{S\in\Ind_{\max}(G_\Mcal),\,f\in\Bset_S},\\
				(\mathbm{v}_{S,f})_i&=\begin{cases}f(P_i),&P_i\in S,\\0,&P_i\notin S,\end{cases}
			\end{aligned}
			\label{eq:vrep}
		\end{equation}
		with $\Ind_{\max}(G_\Mcal)$ the maximal independent sets of the frustration graph and $\Bset_S$ the admissible signs.
		In the worst case, this representation has size $2^{\Theta(\min\{m,n^2\})}$ and is constructible in comparable time.
		Deciding membership in $\STAB(\Mcal)$ is NP-hard~\cite{varela2026predictingmagicmeasurements}.
		Therefore, exact evaluation of $\RoM_\Mcal(\rho)$ cannot be polynomial in $n$ and $m$ for arbitrary instances unless $\mathrm P=\mathrm{NP}$.
		The structural conditions in the main text isolate a subclass with polynomial exact evaluation.
		
		The dependency structure is most transparent in the binary symplectic representation.
		An $n$-qubit Pauli operator modulo phase is a vector $\mathbm{p}=(\mathbm{x}\mid \mathbm{z})\in\Ftwo^{2n}$ with $P(\mathbm{p})\propto\bigotimes_j X^{x_j}Z^{z_j}$.
		Two Pauli operators anticommute iff their symplectic inner product $[\mathbm{p},\mathbm{q}]=\mathbm{x}\cdot \mathbm{z}'+\mathbm{z}\cdot \mathbm{x}'=\mathbm{p}^\top J \mathbm{q}$ is $1$, where $J=\left(\begin{smallmatrix}0&I_n\\I_n&0\end{smallmatrix}\right)$.
		Collecting the vectors of $\Mcal$ as columns of $M_\Mcal\in\Ftwo^{2n\times m}$, the symplectic Gram matrix $\Omega_\Mcal=M_\Mcal^\top J M_\Mcal$ carries the pairwise commutation data, and $G_\Mcal$ is the support graph of its off diagonal entries.
		Every nonzero vector in $\ker_{\Ftwo}(M_\Mcal)$ is a phase-free dependency, and its inclusion-minimal supports are the hyperedges of $\Hcal_\Mcal$.
		For a noncommuting support, an ordered product of Pauli operator representatives can have phase $\pm1$ or $\pm i$ and its phase can depend on the ordering.
		Only an active, commuting dependency has the order-independent $\pm1$ syndrome used below.
		
		Let $S=\{P_i:i\in I\}$ be a commuting context with support submatrix $M_S$.
		A kernel vector $\mathbm{r}\in\ker_{\Ftwo}(M_S)$ encodes a Pauli product relation
		\begin{equation}
			\prod_{i\in I}P_i^{r_i}=(-1)^{\sigma(\mathbm{r})}\mathbb{1},
			\label{eq:relation}
		\end{equation}
		with syndrome $\sigma(\mathbm{r})\in\Ftwo$ set by the product phase.
		As the $P_i$ in $S$ pairwise commute and satisfy $P_i^2=\mathbb{1}$, the left hand side of Eq.~\eqref{eq:relation} squares to $\mathbb{1}$, hence the phase is $\pm1$ and $\sigma(\mathbm{r})$ is well defined.
		Writing $f(P_i)=(-1)^{s_i}$, consistency with every relation requires $\mathbm{r}\cdot \mathbm{s}=\sigma(\mathbm{r})$.
		Under this identification, the admissible sign set is isomorphic to the affine code
		\begin{equation}
			\Bset_S\cong\{\mathbm{s}\in\Ftwo^{|S|}:\ \mathbm{r}\cdot \mathbm{s}=\sigma(\mathbm{r})\ \ \forall\,\mathbm{r}\in\ker_{\Ftwo}(M_S)\}.
			\label{eq:sign-code}
		\end{equation}
		Its parity checks are the dependency relations supported in $S$.
		The inclusion-minimal ones are the dependencies of Definition~\ref{def:active}.
		These sign dependencies are invisible to $G_\Mcal$ alone.
		
		\section{Relaxation and perfect graph solution}
		\label{app:proof-mwis}

		We first characterize the admissible sign sets and prove the exactness criterion, then derive the relaxation of the reduced stabilizer polytope.
		
		\begin{lemmaCollapseRestatement}
			
			$\Mcal$ has no active dependencies if and only if $\Bset_S=\mathsf{Sign}(S)$ for every commuting context $S\in\Ind(G_\Mcal)$.
		\end{lemmaCollapseRestatement}
		
		\begin{proof}
			It suffices to prove the equivalent negated statement that
			$\Mcal$ has an active dependency if and only if
			$\Bset_S\neq\mathsf{Sign}(S)$ for some commuting context
			$S\in\Ind(G_\Mcal)$. If $T$ is an active dependency with
			$\prod_{P\in T}P=\eta\mathbb{1}$, where $\eta\in\{\pm1\}$,
			then $T$ is itself a commuting context.
			The resulting parity check
			$\prod_{P\in T}f(P)=\eta$ excludes half of the sign patterns,
			hence $\Bset_T\neq\mathsf{Sign}(T)$.
			
			Conversely, suppose
			$\Bset_S\neq\mathsf{Sign}(S)$ for some commuting context $S$.
			By Eq.~\eqref{eq:sign-code}, there is a nonzero Pauli product
			relation supported in $S$.
			Choosing one with inclusion-minimal support gives a dependency
			$T\subseteq S$.
			Since all Pauli operators in $S$ pairwise commute, $T$ is active.
			Negating this equivalence proves the lemma.
		\end{proof}
		
		\begin{theoremMwisRestatement}
			
			The equality $\STAB(\Mcal)=\widetilde{\STAB}(\Mcal)$ holds if and only if $\Mcal$ has no active dependencies.
			Under these equivalent conditions, $\RoM_\Mcal(\rho)=\widetilde{\RoM}_\Mcal(\rho)$ for every state $\rho$.
		\end{theoremMwisRestatement}
		
		\begin{proof}
			Equation~\eqref{eq:vrep} gives the exact vertex set $\{\mathbm{v}_{S,f}\}_{S\in\Ind_{\max}(G_\Mcal),\,f\in\Bset_S}$ of $\STAB(\Mcal)$~\cite{varela2026predictingmagicmeasurements}. If $\Mcal$ has no active dependencies, Lemma~\ref{lem:collapse} gives $\Bset_S=\mathsf{Sign}(S)$ for every commuting context. These vertices therefore range over all sign patterns on each maximal independent set, and the exact and relaxed vertex representations coincide.
			Hence $\STAB(\Mcal)=\widetilde{\STAB}(\Mcal)$.
			
			Conversely, let $T$ be an active dependency satisfying $\prod_{P\in T}P=\eta\mathbb{1}$, where $\eta\in\{\pm1\}$, and let $S\in\Ind_{\max}(G_\Mcal)$ contain $T$.
			Choose $\bar f\in\mathsf{Sign}(S)$ such that $\prod_{P\in T}\bar f(P)=-\eta$. Then $\bar f\notin\Bset_S$, whereas $\mathbm{v}_{S,\bar f}\in\widetilde{\STAB}(\Mcal)$.
			Assume that $\mathbm{v}_{S,\bar f}\in\STAB(\Mcal)$, and there exist $\lambda_k>0$, $S_k\in\Ind_{\max}(G_\Mcal)$, and $f_k\in\Bset_{S_k}$ with $\sum_k\lambda_k=1$ such that
			\begin{equation*}
				\mathbm{v}_{S,\bar f}=\sum_k\lambda_k\mathbm{v}_{S_k,f_k}.
			\end{equation*}
			For every $P\in S$, the corresponding coordinate on the left is $\bar f(P)\in\{\pm1\}$, whereas $(\mathbm{v}_{S_k,f_k})_P\in[-1,1]$ for every $k$. Equality in the convex combination therefore requires $(\mathbm{v}_{S_k,f_k})_P=\bar f(P)$ for every $k$ and every $P\in S$.
			Consequently, $S\subseteq S_k$ for every $k$. The maximality of $S$ implies $S_k=S$, and equality on $S$ then gives $f_k=\bar f$, contradicting $f_k\in\Bset_S$ and $\bar f\notin\Bset_S$.
			Thus $\mathbm{v}_{S,\bar f}\in\widetilde{\STAB}(\Mcal)\setminus\STAB(\Mcal)$, proving $\STAB(\Mcal)\subsetneq\widetilde{\STAB}(\Mcal)$ whenever an active dependency exists.
			
			Under the equivalent no active dependency conditions, equality of the polytopes gives $\RoM_\Mcal(\rho)=\widetilde{\RoM}_\Mcal(\rho)$ for every state $\rho$.
			The exact graph dual then follows from Proposition~\ref{prop:relaxation}.
		\end{proof}
		
		\begin{propRelaxationRestatement}
			
			For any Pauli measurement set $\Mcal=\{P_1,\dots,P_m\}$ and a quantum state $\rho$,
			\begin{equation}
				\widetilde{\RoM}_\Mcal(\rho)=\max_{\mathbm{y},\mu}\Big\{\mathbm{b}_{\Mcal(\rho)}^\top \mathbm{y}+\mu: \alpha_{|\mathbm{y}|}(G_\Mcal)+|\mu|\le 1\Big\},
				\label{eq:dual-relaxed-app}
			\end{equation}
			where $|\mathbm{y}|:=(|y_1|,\dots,|y_m|)^\top$. Testing a candidate dual point requires evaluating $\alpha_{|\mathbm{y}|}(G_\Mcal)$. If the inequality fails, a maximizing independent set and its aligned sign pattern identify a violated constraint, so one MWIS computation provides a separation oracle for the displayed feasible region.
		\end{propRelaxationRestatement}
		
		\begin{proof}
			Write $\mathbm{b}=\mathbm{b}_{\Mcal(\rho)}$ for simplicity.
			The primal in Eq.~\eqref{eq:relaxed-rom} is feasible for every $\mathbm{b}$ (the relaxed vertices affinely span $\R^m$) and bounded below ($\|\mathbm{x}\|_1\ge\sum_{\tilde v}x_{\tilde v}=1$).
			Strong linear programming duality therefore gives~\cite{hamaguchi2024handbook}
			\begin{equation}
				\widetilde{\RoM}_\Mcal(\rho)
				=\max_{\mathbm{y},\mu}\Big\{\mathbm{b}^\top\mathbm{y}+\mu:\
				\big|\mathbm{v}_{S,f}^\top\mathbm{y}+\mu\big|\le 1, \forall S\in\Ind_{\max}(G_\Mcal),f\in\mathsf{Sign}(S)\Big\}.
			\end{equation}
			For fixed $S$ and $(\mathbm{y},\mu)$, the worst sign pattern gives
			\begin{equation}
				\max_{f\in\mathsf{Sign}(S)}\Big|\sum_{P_i\in S}f(P_i) y_i+\mu\Big|=\sum_{P_i\in S}|y_i|+|\mu|,
				\label{eq:relaxed-sep}
			\end{equation}
			Choose $s=\sgn(\mu)$ when $\mu\neq0$ and either $s\in\{\pm1\}$ when $\mu=0$; set $f(P_i)=s\,\sgn(y_i)$ for $y_i\neq0$, with arbitrary signs on zero-weight coordinates.
			For nonnegative weights $|y_i|$, maximizing the right hand side over maximal independent sets is equivalent to maximizing over all independent sets, hence
			\begin{equation}
				\max_{S\in\Ind_{\max}(G_\Mcal)}\sum_{P_i\in S}|y_i|=\alpha_{|\mathbm{y}|}(G_\Mcal).
			\end{equation}
			Thus the dual constraints hold for all $(S,f)$ iff $\alpha_{|\mathbm{y}|}(G_\Mcal)+|\mu|\le1$, which is Eq.~\eqref{eq:dual-relaxed-app}.
		\end{proof}
		
		In the absence of active dependencies, Proposition~\ref{prop:relaxation} and Theorem~\ref{thm:mwis} reduce the exact problem to
		\begin{equation}
			\RoM_\Mcal(\rho)=\max_{\mathbm{y},\mu}\Big\{\mathbm{b}_{\Mcal(\rho)}^\top \mathbm{y}+\mu:\ \alpha_{|\mathbm{y}|}(G_\Mcal)+|\mu|\le 1\Big\}.
			\label{eq:graph-program}
		\end{equation}
		
		For a graph $G=(V,E)$ and nonnegative weights $\mathbm{w}$, define the weighted Lov\'asz number by
		\begin{equation}\label{eq:theta-sdp}
			\vartheta(G,\mathbm{w}):=\max\Big\{\textstyle\sum_{i,j}\sqrt{w_iw_j}\,X_{ij}:\ X\succeq0,
			\tr (X)=1,\ X_{ij}=0,\ \forall\,\{i,j\}\in E\Big\}.
		\end{equation}
		Let $\mathcal{Q}(G)$ denote the family of all cliques of $G$. A fractional clique cover assigns a nonnegative weight $\lambda_Q$ to every $Q\in\mathcal{Q}(G)$. Each node $v$ receives the total weight of the cliques containing it, and this total weight is at least its prescribed weight $w_v$. The weighted fractional clique cover number is the minimum total clique weight defined as
		\begin{equation}
			\bar\chi_{\mathbm{w}}(G)
			:=\min_{\lambda_Q\ge0}\Big\{
			\sum_{Q\in\mathcal{Q}(G)}\lambda_Q:\
			\sum_{\substack{Q\in\mathcal{Q}(G)\\v\in Q}}\lambda_Q\ge w_v, 
			\forall v\in V\Big\}.
			\label{eq:fractional-clique-cover}
		\end{equation}
		
		\begin{lemma}[Lov\'asz sandwich\cite{lovasz1979shannon,grotschel1988geometric}]
			\label{lem:lovasz-sandwich}
			For every finite graph $G$ and every $\mathbm{w}\in\mathbb{R}_{\ge0}^{V}$,
			\begin{equation}
				\alpha_{\mathbm{w}}(G)\le\vartheta(G,\mathbm{w})\le\bar\chi_{\mathbm{w}}(G).
				\label{eq:sandwich}
			\end{equation}
		\end{lemma}
		
		\begin{proof}[Proof]
			Every independent set yields a feasible rank one matrix in Eq.~\eqref{eq:theta-sdp} whose objective equals its weight.
			The semidefinite programming (SDP) dual is bounded by every fractional clique cover, giving the second inequality.
			The detailed proof is referred to\cite{lovasz1979shannon,grotschel1988geometric}.
		\end{proof}
		
		\begin{lemma}[Perfect graph collapse~\cite{grotschel1981ellipsoid}]
			\label{lem:perfect-collapse}
			If $G$ is perfect, then for every $\mathbm{w}\in\mathbb{R}_{\ge0}^{V}$,
			\begin{equation}
				\alpha_{\mathbm{w}}(G)=\vartheta(G,\mathbm{w})=\bar\chi_{\mathbm{w}}(G).
				\label{eq:perfect-collapse}
			\end{equation}
			Moreover, with the empty clique included,
			\begin{equation}
				\big\{\mathbm{z}\ge0:\alpha_{\mathbm{z}}(G)\le1\big\}
				=\conv\big\{\mathbm{1}_Q:\ Q\text{ is a clique of }G\big\}.
				\label{eq:antiblocker}
			\end{equation}
			For rational input weights, a MWIS and its value are computable in time polynomial in the input size.
		\end{lemma}
		
		\begin{proof}[Proof]
			For a perfect graph, weighted stable set and fractional clique cover duality gives $\alpha_{\mathbm{w}}(G)=\bar\chi_{\mathbm{w}}(G)$~\cite{lovasz1972normal,lovasz1972characterization,chvatal1975certain}. Lemma~\ref{lem:lovasz-sandwich} then fixes $\vartheta(G,\mathbm{w})$ between equal endpoints, proving Eq.~\eqref{eq:perfect-collapse}.
			To obtain Eq.~\eqref{eq:antiblocker}, suppose $\alpha_{\mathbm{z}}(G)\le1$. Equation~\eqref{eq:perfect-collapse} gives a fractional clique cover $\{\lambda_Q\}$ with $\sum_Q\lambda_Q\le1$ and $\mathbm{z}\le\sum_Q\lambda_Q\mathbm{1}_Q$. The convex hull of clique incidence vectors, including the empty clique, is downward closed, hence it contains $\mathbm{z}$.
			Conversely, every independent set meets a clique in at most one node. Any convex combination of clique incidence vectors therefore satisfies $\alpha_{\mathbm{z}}(G)\le1$. This proves Eq.~\eqref{eq:antiblocker}, equivalently the perfect-graph antiblocker identity of Fulkerson~\cite{fulkerson1971blocking,fulkerson1972antiblocking}.
			The weighted algorithm follows from the optimization-separation equivalence in Theorem~3.1 and the perfect graph construction in Sec.~6, pp.~192--194, of Gr\"otschel, Lov\'asz, and Schrijver~\cite{grotschel1981ellipsoid}. Continuity extends the equality from rational to real nonnegative weights.
		\end{proof}
		
		\begin{theoremSolvableRestatement}
			If $\Mcal$ has no active dependencies and $G_\Mcal$ is perfect, then for any $n$-qubit quantum state $\rho$,
			\begin{equation}
				\RoM_\Mcal\big(\rho\big)
				=
				\max\!\Big(1,\ \max_{Q} \sum_{P\in Q}\big|\tr(P\rho)\big| \Big)
				\le\sqrt{\clnum(G_\Mcal)},
				\label{eq:closed-form-app}
			\end{equation}
			where $Q$ ranges over cliques of $G_\Mcal$, equivalently over pairwise anticommuting subsets of $\Mcal$.
			The bound is attained by any state supported on the $+1$ eigenspace of $|Q|^{-1/2}\sum_{P\in Q}P$ for a maximum anticommuting clique $Q$.
		\end{theoremSolvableRestatement}
		
		\begin{proof}
			Set $G=G_\Mcal$, $\mathbm{b}=\mathbm{b}_{\Mcal(\rho)}$, and
			\begin{equation}
				K(\mathbm{b}):=\max_{Q\ \mathrm{clique}}\sum_{P\in Q}|b_P|.
			\end{equation}
			By Theorem~\ref{thm:mwis} and Proposition~\ref{prop:relaxation}, Eq.~\eqref{eq:graph-program} is the exact dual program.
			Positive homogeneity of $\alpha_{\mathbm{z}}(G)$ and Eq.~\eqref{eq:antiblocker} give, for every $s\ge0$,
			\begin{equation}
				\{\mathbm{z}\ge0:\alpha_{\mathbm{z}}(G)\le s\}=s\cdot\conv\{\mathbm{1}_Q:Q\text{ is a clique of }G\}.
			\end{equation}
			For a feasible pair $(\mathbm{y},\mu)$, set $s=1-|\mu|$. Since $\alpha_{|\mathbm{y}|}(G)\ge0$, condition $\alpha_{|\mathbm{y}|}(G)+|\mu|\le1$ gives $0\le\alpha_{|\mathbm{y}|}(G)\le s$. The scaled antiblocker identity therefore provides coefficients $\lambda_Q\ge0$ with $\sum_Q\lambda_Q=1$ such that
			\begin{equation}
				|\mathbm{y}|=s\sum_Q\lambda_Q\mathbm{1}_Q.
			\end{equation}
			It follows that
			\begin{equation}
				\mathbm{b}^\top\mathbm{y}+\mu\le|\mathbm{b}|^\top|\mathbm{y}|+|\mu|=s\sum_Q\lambda_Q\sum_{P\in Q}|b_P|+|\mu|\le sK(\mathbm{b})+|\mu|\le\max\big(1,K(\mathbm{b})\big).
			\end{equation}
			The final inequality uses $s+|\mu|=1$, making $sK(\mathbm{b})+|\mu|$ a convex combination of $K(\mathbm{b})$ and $1$.
			Conversely, $(\mathbm{y},\mu)=(\mathbm{0},1)$ is feasible and attains $1$.
			For any clique $Q$, define $y_P^{(Q)}=\sgn(b_P)$ for $P\in Q$ and $y_P^{(Q)}=0$ otherwise.
			Every independent set meets $Q$ in at most one node, so $\alpha_{|\mathbm{y}^{(Q)}|}(G)\le1$.
			Thus $(\mathbm{y}^{(Q)},0)$ is feasible and attains $\sum_{P\in Q}|b_P|$.
			The upper and lower bounds coincide, proving the equality in Eq.~\eqref{eq:closed-form-app}.
			
			For each clique $Q=\{P_1,\dots,P_{|Q|}\}$, take the restricted expectation vector $\mathbm{b}_Q=(b_{P_1},\dots,b_{P_{|Q|}})^\top\in\mathbb{R}^{|Q|}$. If $\mathbm{b}_Q=\mathbm{0}$, then $\|\mathbm{b}_Q\|_2\le1$ is immediate. Otherwise, denote the normalized vector $\bar{\mathbm{b}}_Q=\mathbm{b}_Q/\|\mathbm{b}_Q\|_2$, whose components are $\bar b_P=b_P/\|\mathbm{b}_Q\|_2$. We define a Hermitian operator
			\begin{equation}
				\displaystyle
				A_Q(\bar{\mathbm{b}}_Q):=\sum_{P\in Q}\bar b_P P.
			\end{equation}
			The normalized coefficients satisfy $\sum_{P\in Q}\bar b_P^2=1$. Since $P^2=\mathbb{1}$ for every $P\in Q$ and distinct operators in $Q$ anticommute, we obtain
			\begin{equation}
				\displaystyle
				A_Q(\bar{\mathbm{b}}_Q)^2
				=\sum_{P\in Q}\bar b_P^2P^2
				+\sum_{\{P,R\}\subseteq Q}\bar b_P\bar b_R(PR+RP)
				=\left(\sum_{P\in Q}\bar b_P^2\right)\mathbb{1}
				=\mathbb{1}.
			\end{equation}
			By construction, the expectation value of this Hermitian operator is
			\begin{equation}
				\displaystyle
				\tr\!\big(A_Q(\bar{\mathbm{b}}_Q)\rho\big)
				=\sum_{P\in Q}\bar b_P\tr(P\rho)
				=\sum_{P\in Q}\frac{b_P^2}{\|\mathbm{b}_Q\|_2}
				=\|\mathbm{b}_Q\|_2.
			\end{equation}
			The variance of a Hermitian operator in a quantum state is nonnegative. Using $A_Q(\bar{\mathbm{b}}_Q)^2=\mathbb{1}$ and $\tr(\rho)=1$, we obtain
			\begin{equation}
				\displaystyle
				\operatorname{Var}_{\rho}\!\big[A_Q(\bar{\mathbm{b}}_Q)\big]
				=\tr\!\big(A_Q(\bar{\mathbm{b}}_Q)^2\rho\big)
				-\Big[\tr\!\big(A_Q(\bar{\mathbm{b}}_Q)\rho\big)\Big]^2
				=1-\|\mathbm{b}_Q\|_2^2\geq 0.
			\end{equation}
			Therefore $\sum_{P\in Q}b_P^2=\|\mathbm{b}_Q\|_2^2\le1$~\cite{cabello2014graph,xu2024bounding}. The Cauchy--Schwarz inequality gives
			\begin{equation}
				\displaystyle
				\sum_{P\in Q}|\tr(P\rho)|
				=\sum_{P\in Q}|b_P|
				\le\sqrt{|Q|}\,\|\mathbm{b}_Q\|_2
				\le\sqrt{|Q|}.
			\end{equation}
			This holds for all cliques, so the clique term in Eq.~\eqref{eq:closed-form-app} is at most $\sqrt{\clnum(G_\Mcal)}$. The remaining term $1$ obeys the same bound since a nonempty measurement set gives $G_\Mcal$ at least one vertex and hence $\clnum(G_\Mcal)\ge1$. Substitution into Eq.~\eqref{eq:closed-form-app} gives $\RoM_\Mcal(\rho)\le\sqrt{\clnum(G_\Mcal)}$.
			
			It remains to show that this bound is attainable. Let $Q$ be a maximum clique and set $q:=|Q|=\clnum(G_\Mcal)$. Define
			\begin{equation}
				\displaystyle A_Q:=q^{-1/2}\sum_{P\in Q}P.
			\end{equation}
			The same square-expansion argument gives $A_Q^2=\mathbb{1}$ because the coefficients $q^{-1/2}$ have unit Euclidean norm. The no-active-dependency hypothesis excludes $\mathbb{1}\in\Mcal$, so every operator in $Q$ and hence $A_Q$ is traceless. An involution with zero trace has a nonempty $+1$ eigenspace. Let $\rho_Q$ be any state supported on the $+1$ eigenspace. Then $A_Q\rho_Q=\rho_QA_Q=\rho_Q$.
			For each $P\in Q$, the $P$ term contributes $\{P,P\}=2\mathbb{1}$, whereas every $R\neq P$ contributes $\{P,R\}=0$. Therefore
			\begin{equation}
				\displaystyle
				\{P,A_Q\}
				=q^{-1/2}\{P,P\}
				+q^{-1/2}\sum_{R\in Q\setminus\{P\}}\{P,R\}
				=2q^{-1/2}\mathbb{1}.
			\end{equation}
			Using $A_Q\rho_Q=\rho_QA_Q=\rho_Q$ and cyclicity of the trace,
			\begin{equation}
				\displaystyle
				\tr\!\big(\{P,A_Q\}\rho_Q\big)
				=\tr(PA_Q\rho_Q)+\tr(A_QP\rho_Q)
				=2\tr(P\rho_Q).
			\end{equation}
			Comparing the last two equations gives $\tr(P\rho_Q)=q^{-1/2}$ for every $P\in Q$. Consequently,
			\begin{equation}
				\displaystyle
				\sum_{P\in Q}|\tr(P\rho_Q)|
				=q\,q^{-1/2}
				=\sqrt q
				=\sqrt{\clnum(G_\Mcal)}.
			\end{equation}
			Eq.~\eqref{eq:closed-form-app} therefore attains the upper bound. Every state supported on the $+1$ eigenspace satisfies the same calculation, so the saturating state need not be unique when that eigenspace is degenerate.
		\end{proof}
		
		\begin{proposition}[Universal ceiling]
			\label{prop:extremal}
			If $\Mcal$ defined on $n$-qubit Hilbert space has no active dependencies and $G_\Mcal$ is perfect,
			\begin{equation}
				\max_\Mcal\ \sup_\rho\RoM_\Mcal(\rho)=\sqrt{2n+1}.
				\label{eq:extremal}
			\end{equation}
			Moreover, $\sup_\rho\RoM_\Mcal(\rho)=\sqrt{2n+1}$ if and only if $\clnum(G_\Mcal)=2n+1$, i.e. $\Mcal$ contains a maximum set of $2n+1$ pairwise anticommuting Pauli operators.
			
		\end{proposition}
		
		\begin{proof}
			
			For any measurement set $\Mcal$ in the solvable regime, the upper bound and the saturation statement in Theorem~\ref{thm:solvable} give the witness capacity:
			\begin{equation}
				\displaystyle
				\sup_\rho\RoM_\Mcal(\rho)
				=\sqrt{\clnum(G_\Mcal)}.
			\end{equation}
			It therefore suffices to prove that every clique of an $n$-qubit Pauli frustration graph has at most $2n+1$ nodes, and then to exhibit a solvable measurement set attaining this size.
			
			Let $Q=\{P_1,\ldots,P_k\}$ be a clique of $G_\Mcal$. Write $\mathbm{v}_i\in\Ftwo^{2n}$ for the binary symplectic vector of $P_i$, and collect these vectors as columns of
			\begin{equation}
				\displaystyle
				M_Q:=(\mathbm{v}_1\ \cdots\ \mathbm{v}_k)\in\Ftwo^{2n\times k}.
			\end{equation}
			Let $J$ denote the standard symplectic matrix. The symplectic Gram matrix $\Omega_Q=M_Q^\top J M_Q$ has entries $(\Omega_Q)_{ij}=v_i^\top Jv_j$. Its diagonal entries are $0$ since the symplectic form is alternating, while the off-diagonal entries equal $1$ as operators in $Q$ pairwise anticommute. Hence
			\begin{equation}
				\displaystyle
				\Omega_Q=J_k+\mathbb{1}_k,
			\end{equation}
			where $J_k$ is the $k\times k$ all-ones matrix. Matrix rank cannot increase under multiplication, so
			\begin{equation}
				\displaystyle
				\operatorname{rank}_{\Ftwo}(\Omega_Q)
				\le\operatorname{rank}_{\Ftwo}(M_Q)
				\le2n.
			\end{equation}
			
			We now compute the rank of $J_k+\mathbb{1}_k$. For any $\mathbm{x}\in\Ftwo^k$, let $s:=\mathbm{1}^\top\mathbm{x}\in\Ftwo$. Since $J_k\mathbm{x}=s\mathbm{1}$,
			\begin{equation}
				\displaystyle
				(J_k+\mathbb{1}_k)\mathbm{x}
				=s\mathbm{1}+\mathbm{x}.
			\end{equation}
			If $\mathbm{x}$ lies in the kernel, this equation forces $\mathbm{x}=s\mathbm{1}$. For $s=0$ this gives $\mathbm{x}=\mathbm{0}$. For $s=1$ it gives $\mathbm{x}=\mathbm{1}$, which is consistent with $s=\mathbm{1}^\top\mathbm{x}$ when $k$ is odd. Therefore
			\begin{equation}
				\displaystyle
				\ker_{\Ftwo}(\Omega_Q)
				=
				\begin{cases}
					\{\mathbm{0}\},&k\text{ even},\\
					\{\mathbm{0}, \mathbm{1}\},&k\text{ odd}.
				\end{cases}
			\end{equation}
			The rank-nullity theorem then gives
			\begin{equation}
				\displaystyle
				\operatorname{rank}_{\Ftwo}(\Omega_Q)
				=
				\begin{cases}
					k,&k\text{ even},\\
					k-1,&k\text{ odd}.
				\end{cases}
			\end{equation}
			Combining this identity with $\operatorname{rank}_{\Ftwo}(\Omega_Q)\le2n$ gives $k\le2n$ when $k$ is even and $k\le2n+1$ when $k$ is odd. Thus every clique satisfies $k\le2n+1$, and consequently $\clnum(G_\Mcal)\le2n+1$~\cite{sarkar2021sets}.
			
			We then characterize the equality case. Suppose $k$ is odd and $k=2n+1$. The preceding calculation gives $\operatorname{rank}_{\Ftwo}(\Omega_Q)=k-1=2n$. Together with
			\begin{equation}
				\displaystyle
				2n=\operatorname{rank}_{\Ftwo}(\Omega_Q)
				\le\operatorname{rank}_{\Ftwo}(M_Q)
				\le2n,
			\end{equation}
			leads to $\operatorname{rank}_{\Ftwo}(M_Q)=2n$. Since $M_Q$ has $2n+1$ columns, rank-nullity gives $\dim_{\Ftwo}\ker_{\Ftwo}(M_Q)=1$. Moreover, $M_Q\mathbm{x}=\mathbm{0}$ gives
			\begin{equation}
				\displaystyle
				\Omega_Q\mathbm{x}=M_Q^\top J M_Q\mathbm{x}=\mathbm{0},
			\end{equation}
			i.e. $\ker_{\Ftwo}(M_Q)\subseteq\ker_{\Ftwo}(\Omega_Q)$. Note that $\dim_{\Ftwo}\ker_{\Ftwo}(M_Q)=1$, and $\ker_{\Ftwo}(\Omega_Q)=\{\mathbm{0}, \mathbm{1}\}$. Hence
			\begin{equation}\label{eq: overal dependency}
				\displaystyle
				\ker_{\Ftwo}(M_Q)=\{\mathbm{0}, \mathbm{1}\},
				\qquad
				M_Q\mathbm{1}=\sum_{i=1}^{2n+1}v_i=\mathbm{0}.
			\end{equation}
			Addition of binary symplectic vectors represents multiplication of Pauli operators modulo phase. Eq.\eqref{eq: overal dependency} gives a sign dependency $\prod_{i=1}^{2n+1}P_i\propto\mathbb{1}$ over all $2n+1$ nodes in the clique. Because $\ker_{\Ftwo}(M_Q)$ contains only $\mathbm{0}$ and $\mathbm{1}$, this is the unique nontrivial dependency supported on $Q$. Its full support is pairwise anticommuting rather than commuting, so the dependency is inactive.
			
			The canonical JW set $\Gamma=\{\gamma_i\}_{i=1}^{2n+1}$ realizes $G_\Gamma=K_{2n+1}$~\cite{brauer1935spinors,ipek2026phasespace}. Its unique dependency is $\prod_i\gamma_i\propto\mathbb{1}$ and is inactive by the preceding argument. Since complete graphs are perfect, $\Gamma$ lies in the solvable regime and has $\clnum(G_\Gamma)=2n+1$. It therefore attains Eq.~\eqref{eq:extremal}. Conversely, for any solvable $\Mcal$, the identity $\sup_\rho\RoM_\Mcal(\rho)=\sqrt{\clnum(G_\Mcal)}$ reaches $\sqrt{2n+1}$ exactly when $\clnum(G_\Mcal)=2n+1$, equivalently when $\Mcal$ contains such a maximum anticommuting clique.
			
		\end{proof}

		\section{Squared Pauli profiles and stabilizer entropy}
		\label{app:sre}
		
		For a Pauli measurement set $\Mcal$ and an $n$-qubit state $\rho$, write $x_P(\rho)=\tr(P\rho)^2$ for $P\in\Mcal$.
		For $\Mcal=\{P_1,\ldots,P_m\}$, the associated squared expectation vector is $\mathbm{x}_{\Mcal(\rho)}:=(x_{P_1}(\rho),\ldots,x_{P_m}(\rho))^\top\in\mathbb{R}_{\ge0}^m$. This vector supports two complementary questions. A weighted sum $\sum_Pw_Px_P$ probes concentration along a chosen collection of Pauli directions and leads to a graph optimization, whereas the self-collision $\sum_Px_P^2$ probes the global concentration of the profile and leads to stabilizer R\'enyi entropy. The connection between these two viewpoints is that pure-stabilizer profiles are supported on independent sets of the same frustration graph.
		Squaring removes the eigenvalue signs: the squared profile of a pure stabilizer state is the incidence vector of the Pauli operators in $\Mcal$ that belong, up to sign, to its stabilizer group.
		This support is an independent set of $G_\Mcal$, while every independent set extends to a maximal commuting Pauli group and is contained in the squared profile of one of its joint stabilizer eigenstates.
		Since $\sum_{P\in\Mcal}w_P\tr(P\sigma)^2$ is convex in $\sigma$ for nonnegative weights, it follows that, for every $\mathbm{w}:=(w_{P_1},\ldots,w_{P_m})^\top\in\mathbb{R}_{\ge0}^m$,
		\begin{equation}
			\max_{\sigma\in\STAB_n}\sum_{P\in\Mcal}w_P\tr(P\sigma)^2=\alpha_{\mathbm{w}}(G_\Mcal).
			\label{eq:quadratic-stabilizer}
		\end{equation}
		Equation~\eqref{eq:quadratic-stabilizer} is the first point of contact with the graph framework: the independent-set geometry determines the stabilizer optimum of every weighted squared-expectation functional. Unlike the linear expectation vector used in the reduced robustness, the squared expectation vector retains only the support of a stabilizer assignment and discards its signs. Consequently, Pauli sign dependencies impose no additional constraint on this quadratic optimum.
		
		The corresponding optimization over all quantum states is the weighted beta number
		\begin{equation}
			\beta(G_\Mcal,\mathbm{w}):=\max_\rho\sum_{P\in\Mcal}w_P\tr(P\rho)^2\ge\alpha_{\mathbm{w}}(G_\Mcal),
			\label{eq:weighted-beta}
		\end{equation}
		which is independent of the Pauli realization of $G_\Mcal$~\cite{xu2024bounding,xu2025simultaneous}.
		A graph is $\hbar$-perfect when equality holds in Eq.~\eqref{eq:weighted-beta} for every nonnegative weight vector; every perfect graph is $\hbar$-perfect~\cite{xu2025simultaneous}.
		Consequently, a strict inequality in Eq.~\eqref{eq:weighted-beta} certifies nonstabilizerness, whereas an $\hbar$-perfect graph admits no witness from this quadratic family.
		Thus $\alpha_{\mathbm{w}}$ and $\beta$ compare the largest stabilizer and quantum concentrations of the squared expectation vector in a prescribed nonnegative direction $\mathbm{w}$. This directional comparison should be distinguished from the self-collision of that vector considered below.
		This graph collapse is parallel to, but distinct from, contextuality in the Cabello--Severini--Winter framework: a perfect exclusivity graph has coincident graph stable-set polytope $\stab(G)$ and quantum theta body $\mathrm{TH}(G)$, while edges of $G_\Mcal$ encode Pauli anticommutation rather than exclusive events~\cite{cabello2014graph}.
		Perfect frustration graphs therefore do not imply a noncontextual Pauli subtheory; instead, together with the absence of active dependencies, they give a tractable measurement structure for detecting magic, which contextuality supplies as a necessary resource for magic-state quantum computation in the settings studied by Howard \emph{et al.}~\cite{howard2014contextuality}.
		
		We now specialize to the complete Pauli set $\Pcal_n=\{\mathbb{1},X,Y,Z\}^{\otimes n}$ and write $d=2^n$.
		For any $n$-qubit state $\rho$, there is
		\begin{equation}
			\sum_{P\in\Pcal_n}x_P(\rho)=d\,\tr(\rho^2),
			\label{eq:parseval}
		\end{equation}
		and hence defines the normalized Pauli distribution $\Xi_P(\rho):=x_P(\rho)/[d\,\tr(\rho^2)]$.
		The frustration graph is $G_{\Pcal_n}=K_1\sqcup\mathrm{Sp}(2n,2)$, where $K_1$ is the isolated identity node and $\mathrm{Sp}(2n,2)$ has the nonzero binary Pauli vectors as nodes, adjacent when their symplectic product is one~\cite{brouwer2012spectra}.
		Its maximum independent sets are the maximal abelian Pauli supports, each containing $d$ operators~\cite{sarkar2021sets} that 
		\begin{equation}
			\alpha(G_{\Pcal_n})=d=2^n.
			\label{eq:alpha-full}
		\end{equation}
		For a pure stabilizer state, $\Xi$ is therefore the uniform distribution on one of these maximum independent sets: it equals $1/d$ on the $d$ Pauli operators in the stabilizer support and vanishes elsewhere. Conversely, a pure-state Pauli distribution of this form identifies a joint eigenstate of a maximal abelian Pauli group and hence a pure stabilizer state. The family of maximum independent sets thus specifies not only the stabilizer value $d$, but also the extremal distributional shape against which the collision of a general pure-state Pauli profile can be compared.
		In particular, the all-one weight vector gives
		\begin{equation}
			\beta(G_{\Pcal_n},\mathbm{1})=\max_\rho d\,\tr(\rho^2)=d=\alpha(G_{\Pcal_n}),
			\label{eq:beta-uniform-full}
		\end{equation}
		that uniform weights cannot reveal $\hbar$-imperfection.
		
		The $\hbar$-perfectness of the full Pauli graph is dimension dependent.
		For one qubit its nonidentity component is $K_3$ and is perfect; for two qubits the $15$-node nonidentity graph is $\hbar$-perfect although it is not perfect~\cite{xu2025simultaneous}.
		For $n\ge3$, $G_{\Pcal_n}$ contains the $\hbar$-imperfect anti-heptagon as an induced subgraph, for example through the three-qubit Pauli operators $\{IIX,IXI,IIY,IYY,XZZ,YYZ,ZYX\}$, and $\hbar$-perfectness is inherited by induced subgraphs~\cite{xu2024bounding,xu2025simultaneous}.
		
		The stabilizer R\'enyi entropy is nonlinear in the squared profile and depends on its self-collision $\sum_Px_P^2$: profiles concentrated on fewer Pauli directions have larger collision, whereas profiles spread over more directions have smaller collision. For pure states, the maximum-independent-set distributions identified above provide the stabilizer reference scale for this comparison.
		For a pure state $\psi=\ket{\psi}\!\bra{\psi}$, define
		\begin{equation}
			A_2(\psi):=\frac1d\sum_{P\in\Pcal_n}\tr(P\psi)^4=\frac1d\sum_Px_P(\psi)^2,
			\qquad M_2(\psi):=-\log A_2(\psi)
			\label{eq:sre-definition}
		\end{equation}
		as in Ref.~\cite{leone2022stabilizer}.
		Equation~\eqref{eq:parseval} gives $\sum_Px_P(\psi)=d$, and $x_P^2\le x_P$ gives $\sum_Px_P(\psi)^2\le d$.
		Equality holds precisely when the squared profile is the indicator of a maximal abelian Pauli support, equivalently when $\psi$ is a pure stabilizer state.
		With the collision entropy $H_2(\Xi):=-\log\sum_P\Xi_P^2$, Eqs.~\eqref{eq:alpha-full} and~\eqref{eq:sre-definition} yield
		\begin{equation}
			\begin{aligned}
				M_2(\psi)
				&=H_2(\Xi)-\log\alpha(G_{\Pcal_n})\\
				&=-\log\!\left[\alpha(G_{\Pcal_n})\sum_{P\in\Pcal_n}\Xi_P(\psi)^2\right],
			\end{aligned}
			\label{eq:sre-pure}
		\end{equation}
		Hence $\alpha(G_{\Pcal_n})^{-1}$ is the collision probability of a uniform maximum-independent-set distribution. Equation~\eqref{eq:sre-pure} measures the collision-entropy excess of $\Xi(\psi)$ relative to this graph-defined stabilizer scale. In particular, $M_2(\psi)=0$ exactly when $\Xi(\psi)$ is uniform on a maximum independent set of $G_{\Pcal_n}$, equivalently when $\psi$ is a pure stabilizer state. The graph enters this interpretation through the family and size of the extremal stabilizer supports; $M_2$ itself is not asserted to be an invariant determined by the abstract graph alone.
		
		For a mixed state $\rho$, define $A_2(\rho):=d^{-1}\sum_P\tr(P\rho)^4$, $S_2(\rho):=-\log\tr(\rho^2)$, and $\widetilde M_2(\rho):=-\log[A_2(\rho)/\tr(\rho^2)^3]$.
		The definition of $\Xi$ and Eq.~\eqref{eq:parseval} give the exact identity
		\begin{equation}
			\widetilde M_2(\rho)=H_2(\Xi)-\log\alpha(G_{\Pcal_n})-S_2(\rho),
			\label{eq:sre-mixed}
		\end{equation}
		which reduces to Eq.~\eqref{eq:sre-pure} for pure states.
		This is an algebraic decomposition rather than a nonnegative mixed-state magic measure; for example, $\widetilde M_2(\mathbb{1}/d)=-2\log d$.
		The additional term $-S_2(\rho)$ makes explicit why the pure-state support interpretation does not extend directly: spreading of the Pauli distribution can arise from ordinary mixing as well as from nonstabilizerness.
		
		For a maximum independent set $\Scal$ of $G_{\Pcal_n}$, let $\Delta_\Scal(\rho):=d^{-1}\sum_{P\in\Scal}\tr(P\rho)P$ denote dephasing in its joint stabilizer basis.
		Then $\sum_{P\in\Scal}x_P(\rho)=d\,\tr[\Delta_\Scal(\rho)^2]$, and dephasing cannot increase purity, giving
		\begin{equation}
			\alpha_{\mathbm{x}_{\Pcal_n(\rho)}}(G_{\Pcal_n})=d\max_{\Scal\in\Ind_{\max}(G_{\Pcal_n})}\tr[\Delta_\Scal(\rho)^2]\le d\,\tr(\rho^2).
			\label{eq:purity-relax}
		\end{equation}
		Thus the state-dependent MWIS restores a direct graph optimization: it selects the stabilizer basis retaining the largest dephased purity. This quantity and $\widetilde M_2$ are complementary summaries of the same squared profile, but Eq.~\eqref{eq:purity-relax} does not turn the mixed-state moment identity in Eq.~\eqref{eq:sre-mixed} into a stabilizer bound. 
		Establishing such a bound is still a missing step for a graph theoretic mixed-state resource criterion.
		
		\makeatletter
		\let\NAT@bibsetnum\NATx@bibsetnum
		\let\app@label\label
		\def\label#1{%
			\def\app@labelname{#1}%
			\def\app@lastbibitem{LastBibItem}%
			\ifx\app@labelname\app@lastbibitem
			\app@label{AppendixLastBibItem}%
			\else
			\app@label{#1}%
			\fi}
		\makeatother
		
		\renewcommand{\refname}{}
		
		\makeatletter
		\immediate\write\@bibunitaux{\string\citation{BFLongBibliographyControl}}
		\makeatother
		\putbib[references]
	\end{bibunit}
	
\end{document}